\documentclass[11pt]{article}
\usepackage[latin1]{inputenc}
\usepackage[small]{titlesec}
\usepackage{amsmath}
\usepackage{amsthm}
\usepackage{amssymb}
\usepackage{graphicx}
\usepackage{standalone}
\usepackage{color}
\theoremstyle{plain}
\newtheorem{thm}{Theorem}[section]

\newtheorem{prop}[thm]{Proposition}
\newtheorem{cor}[thm]{Corollary}
\newtheorem{res}[thm]{Result}

\newtheorem{defn}[thm]{Definition}
\newtheorem{exmp}[thm]{Example}
\usepackage[small]{titlesec}
\usepackage{blindtext}

	\title{Bent functions and strongly regular graphs}
\author{Valentino Smaldore
\footnote{Dipartimento di Tecnica e Gestione dei Sistemi Industriali,
Universit\`{a} degli Studi di Padova, Stradella S. Nicola 3, 36100 Vicenza, Italy.
}}

\date{}

\begin{document}
\maketitle

\begin{abstract}
 The family of bent functions is a known class of Boolean functions, which have a great importance in cryptography. The Cayley graph defined on $\mathbb{Z}_{2}^{n}$ by the support of a bent function is a strongly regular graph $srg(v,k,\lambda,\mu)$, with $\lambda=\mu$. In this note we list the parameters of such Cayley graphs. Moreover, a condition is given on $(n,m)$-bent functions $F=(f_1,\ldots,f_m)$, involving the support of their components $f_i$, and their $n$-ary symmetric differences.
\end{abstract}

\section{Introduction}
 A \textit{cryptosystem} is an encryption and decryption algorithm for a message. If Alice wants to send a message $p$ to Bob, the encryption algorithm $E$ computes the \textit{ciphertext} $z$ starting from a \textit{key} $K_A$, i.e. $z=E(p,K_A)$. Bob uses the decryption algorithm $D$ to recover $p$ from a key $K_B$, i.e. $p=D(z,K_B)$. Necessarily, for all $p,K_A,K_B$, $D(E(p,K_A),K_B)=p$. Cryptosystems are called \textit{private key}, if the parties know each other and have shared information about their private keys, or \textit{public key} if it is not necessary that the two parties know each other, and they have two public keys. The best known private key algorithms are $DES$ (Data Encryption Standard) and its successor $AES$ (Advanced Encryption Standard). The reader can find more information on cryptography in \cite{CRYPTO}. One of the most important features of cryptographic algorithms is the \textit{confusion}, i.e. the relation between any bit and all the plaintext appearing at random. After the linear cryptanalysis techniques of M. Matsui \cite{Matsui}, one of the research items in cryptography was to find functions as far as possible from the linear functions, that is maximizing the Hamming distance, in order to resist to linear attacks, see \cite{BS}. Among the family of Boolean functions, such functions are called \textit{bent functions}. In \cite{srg1, srg2} a characterization of bent functions is given in terms of strongly regular graphs. Here, we give considerations on parameters of such strongly regular graphs, and a first characterization of $(n,m)$-bent functions.

\section{Preliminaries}
Let $\mathbb{Z}_{2}$ be the binary field. A \textit{Boolean function} is a function $f:\mathbb{Z}_{2}^{n}\longrightarrow\mathbb{Z}_{2}$
and to denote $f$ we will use two different notations: the \textit{classical notation}, where the input string is given by $n$ binary variables, and the \textit{$2^{n}$-tuple vector representation} $f=(f_{0},f_{1},\ldots, f_{2^{n}-1})$ where $f_i=f(b(i))$ and $b(i)$ is the binary expansion of the integer $i$. We will denote by $\Omega_{f}$ the \textit{support} of $f$, i.e.
$$\Omega_f=\{w\in\mathbb{Z}_{2}^{n}|f(w)\neq 0\}=\{w\in\mathbb{Z}_{2}^{n}|f(w)=1\}.$$
\begin{defn}
Let $l$ be a Boolean function.
 \begin{itemize}
  \item We say that $l$ is a \textit{linear function} if $\forall x,y\in\mathbb{Z}_{2}^{n}$, $l(x+y)=l(x)+l(y)$.
  \item We say that $l$ is an \textit{affine function} if it is a linear function plus a constant in $\mathbb{Z}_{2}$.
 \end{itemize}
 We denote with $\mathcal{A}$ the set of all affine functions
\end{defn}
The \textit{nonlinearity} of a Boolean function $f$ is the minimum Hamming distance between $f$ and an affine function, i.e.
$$Nl(f)=min_{\phi\in\mathcal{A}}|\{x\in\mathbb{Z}_{2}^{n}|f(x)\neq\phi(x)\}|.$$
\begin{defn}\label{bf}
 A Boolean function $f$ is called \textit{bent function} if $Nl(f)=\frac{2^{n}-2^{\frac{n}{2}}}{2}$.
\end{defn}
Note that by Definition \ref{bf} $n$ must be even. Bent functions are also called $PN$ (perfectly nonlinear).
Here we define the \textit{Abstract Fourier Transform} of a Boolean function $f$ as the rational valued function $f^*$ which defines the coefficients of
$f$ with respect to the orthonormal basis of the group characters $Q_w(x)=(-1)^{(w\cdot x)}$, where $"\cdot"$ is the standard inner product and $w\cdot x=\sum_{i=1}^{n}x_{i}w_{i}=Tr^{n}_{1}(wx)$. Then
$$f^*(w)=\frac{\sum_{x\in\mathbb{Z}_{2}^{n}}(-1)^{Tr^{n}_{1}(wx)}f(x)}{2^{n}}.$$
Note that $f^*(b(0))=\frac{|\Omega_{f}|}{2^{n}}$. The \textit{Walsh spectrum} is the set of values of $f^*(w)$. Here we investigate the spectrum in terms of a graph eigenvalue problem.

\section{The Cayley graph $Cay(\mathbb{Z}_{2}^{n},\Omega_{f})$}
\begin{defn}
 Let $\Gamma$ be a group with identity $e$.
 \begin{itemize}
  \item A \textit{Cayley subset}, is a subset $C\subseteq\Gamma$ such that $e\notin C$ and whenever $g\in C$, then $g^{-1}\in C$.
  \item The \textit{Cayley graph} $G=Cay(\Gamma,C)$ of $\Gamma$ with respect to $C$ is the graph whose vertex set is $\Gamma$, when two vertices $g$ and $h$ are adjacent if and only if $gh^{-1}\in C$.
 \end{itemize}
\end{defn}
We modify this definition by dropping the condition $e\notin C$, allowing loops in the Cayley graph.

Consider now the additive group $(\mathbb{Z}_{2}^{n},\oplus)$, where $\oplus$ is the componentwise sum. For all $w\in\mathbb{Z}_{2}^{n}$, $w^{-1}=w$, then each subset of $\mathbb{Z}_{2}^{n}$ is a Cayley subset. We can associate each Boolean function $f$ to the Cayley graph $G_{f}=Cay(\mathbb{Z}_{2}^{n},\Omega_{f})$. The vertex-set $V(G_{f})$ is the whole $\mathbb{Z}_{2}^{n}$, while the edge-set is $E(G_{f})=\{(u,v)\in\mathbb{Z}_{2}^{n}|u\oplus v\in\Omega_{f}\}=\{(u,v)\in\mathbb{Z}_{2}^{n}|f(u\oplus v)= 1\}$.
The graph has $2^{n-dim\langle\Omega_{f}\rangle}$ vertices which are the cosets of $\langle\Omega_{f}\rangle$ in $\mathbb{Z}_{2}^{n}$.
Since eigenvectors of the Cayley graph are exactly the group characters $Q_w(x)=(-1)^{Tr^{n}_{m}(wx)}$, see \cite{Cayley},the following two results give a characterization of the spectrum of $G_{f}$ from the Walsh spectrum of $f$.

\begin{res}\cite[Theorem 1]{srg1}
 The $i$-th eigenvalue $\lambda_{i}$ of the Cayley graph, which corresponds to the eigenvector $Q_{b(i)}$, is given by
 $$\lambda_{i}=\sum_{x\in\mathbb{Z}_{2}^{n}}(-1)^{Tr^{n}_{1}(b(i)x)}f(x)=2^{n}f^*(b(i)).$$
\end{res}

\begin{res}\cite[Proposition 2]{srg1}
 \begin{enumerate}
  \item The largest spectral coefficients is $\lambda_{0}=2^{n}f^*(b(0))=|\Omega_{f}|$, with multiplicity $2^{n-dim\langle\Omega_{f}\rangle}$.
  \item The number of non zero spectral coefficients is the rank of the adjacency matrix of $G_f$.
  \item If $G_f$ is connected, $f$ has a spectral coefficient equal to $-\lambda_{0}$ if and only if its Walsh spectrum is symmetric with respect to 0.
 \end{enumerate}
\end{res}

\section{Strongly regular graphs}
A strongly regular graph with parameters $(v,k,\lambda,\mu)$, denoted by $srg(v,k,\lambda,\mu)$, is a graph with $v$ vertices, each vertex lies on $k$ edges, any two adjacent vertices have $\lambda$ common neighbours and any two non-adjacent vertices have $\mu$ common neighbours. We give now some folklore results on strongly regular graphs, see \cite{brvan} for more details.

\begin{res}
   \label{fundamentalsrg}
    $k(k-\lambda-1)=\mu(v-k-1)$.
   \end{res}

 The spectrum of the adjacency matrix of an $srg(v,k,\lambda,\mu)$ is fully determined by its parameters.
\begin{res}
    A strongly regular graph $G$ with parameters $(v,k,\lambda,\mu)$ has exactly three eigenvalues: $k$, $\theta_{1}$ and $\theta_{2}$ of multiplicity, respectively, $1$, $m_{1}$ and $m_{2}$, where:
    $$\theta_{1}=\frac{1}{2}\big[(\lambda-\mu)+\sqrt{(\lambda-\mu)^{2}+4(k-\mu)}\big],$$
    $$\theta_{2}=\frac{1}{2}\big[(\lambda-\mu)-\sqrt{(\lambda-\mu)^{2}+4(k-\mu)}\big],$$
    $$m_{1}=\frac{1}{2}\Big[(v-1)-\frac{2k-(v-1)(\lambda-\mu)}{\sqrt{(\lambda-\mu)^{2}+4(k-\mu)}}\Big],$$
    $$m_{2}=\frac{1}{2}\Big[(v-1)+\frac{2k-(v-1)(\lambda-\mu)}{\sqrt{(\lambda-\mu)^{2}+4(k-\mu)}}\Big].$$
    We write the spectrum as $k,\theta_{1}^{m_{1}},\theta_{2}^{m_{2}}$.
    On the other hand, we can express the parameters of a strongly regular graph starting from its spectrum
    $$v=1+m_{1}\theta_{1}+m_{2}\theta_{2},$$
    $$\lambda=k+\theta_{1}\theta_{2}+\theta_{1}+\theta_{2},$$
    $$\mu=k+\theta_{1}\theta_{2}=\lambda-\theta_{1}-\theta_{2}.$$
\end{res}

\begin{cor}
 Consider a $srg(v,k,\lambda,\mu)$, with spectrum $k,\theta_{1}^{m_{1}},\theta_{2}^{m_{2}}$. Then $\lambda=\mu$ if and only if $\theta_{1}=-\theta_{2}$.
\end{cor}

\begin{res}\label{secondsrg}
The parameters $\lambda$ and $\mu$ of a $srg(v,k,\lambda,\mu)$ may be derived from its spectrum, since:
 \begin{equation}
  \begin{cases}
   \lambda=k+\theta_{1}+\theta_{2}+\theta_{1}\theta_{2}\\
   \mu=k+\theta_{1}\theta_{2}.
  \end{cases}
 \end{equation}
\end{res}

In \cite{srg1, srg2} a characterization of bent functions is given in a graph theoretical point of view.
\begin{res}\cite[Lemma 12]{srg1}
 If $f$ is a bent function, the graph $G_f$ is a strongly regular graph with $\lambda=\mu$.
\end{res}

\begin{res}\cite[Theorem 3]{srg2}
 Bent functions are the only functions whose associated Cayley graph $G_f$ is a strongly regular graph with $\lambda=\mu$.
\end{res}

\begin{prop}
 The Cayley graph $G_f$ of a bent function is exactly one of the following:
 \begin{itemize}
  \item $srg(2^{n},\frac{2^{n}+2^{\frac{n}{2}}}{2},\frac{2^{n}+2^{\frac{n}{2}}-2^{n-1}}{2},\frac{2^{n}+2^{\frac{n}{2}}-2^{n-1}}{2});$
  \item $srg(2^{n},\frac{2^{n}-2^{\frac{n}{2}}}{2},\frac{2^{n}-2^{\frac{n}{2}}-2^{n-1}}{2},\frac{2^{n}-2^{\frac{n}{2}}-2^{n-1}}{2}).$
 \end{itemize}
\end{prop}

\begin{proof}
 From \cite[Definition 4]{srg1} we know the three eigenvalues $k,\theta_{1},\theta_{2}=-\theta_{1}$ of $G_f$. From \ref{secondsrg} we get the parameters $\lambda$ and $\mu$, while \ref{fundamentalsrg} allows us to compute $v=2^{n}=|\mathbb{Z}_{2}^{n}|$.
\end{proof}

\begin{exmp}
 The first strongly regular graph defined by bent functions are
 \begin{description}
  \item[$n=2$] \begin{itemize}
   \item $srg(4,3,1,1)$, i.e. the complete graph $K_4$.
   \item $srg(4,1,0,0)$, i.e. a trivial strongly regular graph made of 2 disconnected edges.
  \end{itemize}
  \item[$n=4$] \begin{itemize}
   \item $srg(16,10,6,6)$.
   \item $srg(16,10,2,2)$.
  \end{itemize}
  \item[$n=6$] \begin{itemize}
   \item $srg(64,36,20,20)$.
   \item $srg(64,28,12,12)$.
  \end{itemize}
  \item[$n=8$] \begin{itemize}
   \item $srg(256,136,72,72)$.
   \item $srg(256,120,56,56)$.
  \end{itemize}
  \item[$n=10$] \begin{itemize}
   \item $srg(1024,528,272,272)$.
   \item $srg(1024,496,240,240)$.
  \end{itemize}
 \end{description}
 Note that in each case graphs have the parameters of the complements of the affine polar graphs $VO^{\mp}(2n,2)$, which is the graph arising from a quadric $Q$ in the vector space $V=V(2n,2)$ and two points $u,v\in V$ represent adjacent vertices if and only if $Q(u-v)=0$. Note that the quadric is elliptic or hyperbolic while we consider the first or the second example, respectively. See the table of strongly regular graphs in \cite{TAB} for more details.
\end{exmp}

\section{Vectorial bent function}
 Consider now functions $F:\mathbb{Z}_{2}^{n}\longrightarrow\mathbb{Z}_{2}^{m}$, $F(x_{1},\ldots,x_{n})=(f_{1},\ldots,f_{m})$, where for each $i$, $f_{i}:\mathbb{Z}_{2}^{n}\longrightarrow\mathbb{Z}_{2}$. The set of affine vectorial functions $\mathcal{A}_{n,m}$ is defined as in the case $m=1$. We can introduce two different ways to express the nonlinearity of a vectorial Boolean function:
 \begin{equation}
  nl(F)=min_{v\in\mathbb{Z}_{2}^{n}\setminus\{0\}}Nl(F\cdot v)
 \end{equation}
 \begin{equation}
  Nl(F)=min_{\phi\in\mathcal{A}_{n,m}}|\{x\in\mathbb{Z}_{2}^{n}|F(x)\neq\phi(x)\}|
 \end{equation}
 \begin{defn}
  A $(n,m)$-bent function, or \textit{vectorial bent function}, is a function $F=(f_{1},\ldots,f_{m})$ such that $nl(F)=\frac{2^{n}-2^{\frac{n}{2}}}{2}$, or equivalently each linear combination of $f_{1},\ldots,f_{m}$ is a bent function.
 \end{defn}
In order to give graph based properties of $(n,m)$-bent functions we need now to define the set operation \textit{symmetric difference}, which is the equivalent of the logical operation $XOR$.
\begin{defn}
 The symmetric difference between two sets $A$ and $B$ is
 $$A\triangle B=(A\setminus B)\cup(B\setminus A)=(A\cup B)\setminus(A\cap B).$$
\end{defn}
\begin{prop}
The power set of any set $X$ is an elementary abelian $2$-group under the operation of symmetric difference.
\end{prop}

\begin{proof}
 The symmetric difference is commutative and associative:
 \begin{itemize}
  \item $A\triangle B=B\triangle A$;
  \item $(A\triangle B)\triangle C=A\triangle(B\triangle C)$.
 \end{itemize}
 Moreover the empty set is the identity and each element has order two:
 \begin{itemize}
  \item $A\triangle\emptyset=A$;
  \item $A\triangle A=\emptyset$.
 \end{itemize}
\end{proof}
 An elementary abelian $2$-group is also called \textit{Boolean group}, see \cite{Set} for more details.

 The symmetric difference of a collection of sets is made of elements contained in an odd number of sets. The $n$-ary symmetric difference is defined as follows;
 $$\bigtriangleup \mathcal{M}=\Big\{a\in\bigcup \mathcal{M}\Big|\sharp\{A\in M|a\in A\}=2k+1, k\in\mathbb{N}\Big\}.$$

\begin{prop}
 Consider a vectorial Boolean function $F=(f_{1},\ldots,f_{m})$, with $f_i:\mathbb{Z}_{2}^{n}\longrightarrow\mathbb{Z}_{2}$, and let $\Omega_i=\Omega_{f(i)}$ be the support of $f_i$, of $i=1,\ldots,m$. If the function $F$ is $(n,m)$-bent, then the Cayley graphs $Cay(\mathbb{Z}_{2}^{n},\bigtriangleup_{i\in I}\Omega_{i})$ are strongly regular with $\lambda=\mu$ for all index subset $I\subseteq[1,\ldots, m]$.
\end{prop}

\section{Conclusion}
Future works should extend this notions to the case $n$ odd, by taking into account $APN$ (almost perfectly non linear) functions, i.e. functions which are as close as possible to perfect nonlinearity.


\begin{thebibliography}{0}
  \bibitem{srg1}A. Bernasconi, B. Codenotti, \textit{Spectral Analysis of Boolean Functions as a Graph Eigenvalue Problem}, IEEE Transactions on Computers, 1999, 48(3), pp. 345-351.
  \bibitem{srg2}A. Bernasconi, B. Codenotti, J. M. VanderKam, \textit{A Characterization of Bent Functions in terms of Strongly Regular Graphs}, IEEE Transactions on Computers, 2001, 50(9), pp. 984-985.
 \bibitem{BS}E. Biham, A. Shamir, \textit{Differential cryptanalysis of DES-like cryptosystems}, Journal of Cryptology, 1991, 4, pp. 3-72.
 \bibitem{brvan}A. E. Brouwer, H. Van Maldeghem, \textit{Strongly Regular Graphs}, Encyclopedia of Mathematics and its Applications, Cambridge University Press, 2022.
 \bibitem{TAB}A. E. Brouwer, \textit{Parameters of Strongly Regular Graphs}, https://www.win.tue.nl/~aeb/graphs/srg/srgtab.html
 \bibitem{carlet}C. Carlet, C. Ding, J. Yuan, \textit{Linear codes from perfect nonlinear mappings and their secret sharing schemes}, IEEE Transactions on Information Theory, 2005, 51(6), pp. 2089-2102.
 \bibitem{CM}C. Carlet, S. Mesnager, \textit{Four decades of research on bent functions}, Deisgns, Codes amnd Cryptography, 2016, 78, pp. 5-50.
 \bibitem{vecbent}D. Dong, X. Zhang, L. Qu, S. Fu, \textit{A note on vectorial bent functions}, Information Processing Letters, 2013, 113(22-24), pp. 866-870.
     \bibitem{Set}P. Givant, P. Halmos, \textit{Introduction to Boolean Algebras}, Springer, 2009.
 \bibitem{Hadamard}K. J. Horadam, \textit{Hadamard Matrices and Their Applications}, Princeton Universtity Press, 2007.
 \bibitem{Matsui}M. Matsui, \textit{Linear cryptanalysis method for DES cypher}, EUROCRYPT93, LNCS 765, Springer, 1994, pp. 386-397.
 \bibitem{CRYPTO}A. J. Menezes, P. van Oorschot, S. A. Vanstone, \textit{Handbook of Applied Cryptography}, CRC Press, Boca Raton, 1997.
  \bibitem{Bent}S. Mesnager, \textit{Bent Functions. Fundamentals and Results}, Springer, 2016.
  \bibitem{Cayley}P. H. Zieschang, \textit{Cayley graphs of finite groups}, Journal of Algebra, 1988, 118(2), pp. 447-454.
\end{thebibliography}
\end{document}